\documentclass[aps,prd, 10pt, twocolumn,nofootinbib,showpacs]{revtex4-1}
\usepackage{subeqn}
\usepackage[usenames,dvipsnames]{xcolor} 
\usepackage[normalem]{ulem}
\usepackage{amsthm}
\newtheorem{theorem}{Theorem}
\newtheorem{lemma}[theorem]{Lemma}
\newtheorem{corollary}[theorem]{Corollary}

\begin{document}

\title{Classical uncertainty in predicting the future}
\author{Koray D\"{u}zta\c{s}}
\email{koray.duztas@emu.edu.tr}
\affiliation{Physics Department, Eastern Mediterranean  University, Famagusta, North Cyprus via Mersin 10, Turkey}

\begin{abstract}
In this work we scrutinize the deterministic nature of globally hyperbolic space-times from the point of view of an observer. We show that a space-time point $q \in M$ that lies to the future of an observer at $p \in M$, receives signals that are invisible (to be made precise) to the observer at $p$. Part of the initial data on a Cauchy surface, required to predict what happens at $q$, is also invisible to the observer at $p$. Therefore it is not possible for any observer to predict a future event with certainty. The uncertainty increases as one attempts to predict further future. An observer at $p$ can access the entire data to determine what happens at $q$, if and only if $q \in J^-(p)$. Classical uncertainty in prediction is not an intrinsic feature of the events in space-time. It adds up with the usual quantum mechanical uncertainty to limit our ability to predict the future. We also suggest a thought experiment to elucidate the subject.
 \end{abstract}
\maketitle

A space-time $(M,g)$ is said to be globally hyperbolic if the diamonds $J^+(p)\cap J^-(q)$ are compact for all $p,q \in M$, and strong causality holds at every point in $M$. Globally hyperbolic space-times admit Cauchy surfaces whose domain of dependence is the space-time manifold itself. Any signal sent to $p \in M$ is registered on this Cauchy surface. Therefore given initial conditions on this surface, one should be able to predict what happens at any point $p$ that lies to his future. In this sense general relativity is a deterministic theory. 

The deterministic nature of general relativity was hampered by the development of singularity theorems by Penrose and Hawking~\cite{singtheo}. According to these theorems, the formation of singularities is inevitable as a result of gravitational collapse. The presence of a singularity violates the condition of the compactness of the diamonds $J^+(p)\cap J^-(q)$ for the space-time to be globally hyperbolic. Thus, one cannot define a Cauchy surface and predictability breaks down. Cosmic censorship conjecture was proposed to circumvent this obstacle~\cite{ccc}. This conjecture states that the singularities that ensue as a result of gravitational collapse are always surrounded by an event horizon which disables their causal contact with distant observers. Thus, one can ignore the breakdown of causality at space-time singularities, since a distant observer does not encounter any effects propagating out of the singularities. Hawking argued that, there exists a surface about which an observer can only have limited information such as mass, angular momentum and charge when the singularity is covered by an event horizon. This brings a limitation to our ability to predict the future, even if the cosmic censorship conjecture is correct~\cite{breakdown}.

In this work we evaluate an observer's ability to predict the future in a globally hyperbolic space-time. We point out that there exists a surface about which an observer can access absolutely no information, and the data on this surface is required to predict what happens at a future point. Therefore, it is not possible for any observer to predict a future event with certainty. The uncertainty increases as one attempts to predict further future.  We also show that an observer at $p$ can access the entire data to determine what happens at $q$, if and only if $q \in J^-(p)$. We comment on the analogies and the discrepancies with the usual quantum mechanical uncertainty principle. Finally, we suggest a thought experiment to elucidate the subject of classical uncertainty.

We start with a simple lemma that leads crucial results. 

\begin{lemma} \label{invisible}
Let $p,q \in M$ such that $p<q$. Let $r \in J^-(q)\setminus J^-(p)$. Any past extendible causal curve from $r$ to $q$ is entirely contained in $J^-(q)\setminus J^-(p)$, i.e. in $\rm{Ext} (J^-(p))$.
\end{lemma}
\begin{proof} Assume a causal curve from $r$ intersects $J^-(p)$. Then there exists $p_0 \in J^-(p)$ such that $r<p_0<p$. This implies $r<p$, i.e. $r\in J^-(p)$, which is a contradiction.
\end{proof}
If a set of points are contained in $\rm{Ext} (I^-(p))$ (which is equal to $\rm{Ext} (J^-(p))$ for globally hyperbolic space-times), we define this set of points to be \emph{invisible} to an observer at $p$. (We adopt this definition from Galloway and Woolgar~\cite{invisible}). This definition is intuitive in the sense that information flows along causal curves, and no information can reach an observer at $p$ from the region which we define as invisible. Let us assume an observer at  $r \in J^-(q)\setminus J^-(p)$ sends a signal to $q$. This signal is represented by a future directed causal curve that has a past end point at $r$. Since, this signal is invisible to an observer at $p$, he/she cannot include the effect of this signal in his/her calculations to predict what happens at $q$. 

In globally hyperbolic space-times it is customary to work with Cauchy surfaces. The next theorem states that, at every Cauchy surface, there exists a region which is invisible to an observer at $p$. 

\begin{theorem} \label{non-empty}
Let $S$ be any Cauchy surface for the space-time $M$, or a partial Cauchy surface for $p,q$ such that $p<q$. The region $S\cap (J^-(q)\setminus J^-(p))$ is non-empty and invisible to $p$. Thus, the region $S\setminus J^-(p)$ is non-empty and invisible to  $p$. 
\end{theorem}
\begin{proof}
If the sets $S\cap (J^-(q)\setminus J^-(p))$ and $S\setminus J^-(p)$ are non-empty they are trivially invisible to $p$, since they are both contained in $\rm{Ext} (J^-(p))$.

To prove that they are non-empty let us choose $r \in J^-(q)\setminus J^-(p)$ such that $r$ lies to the past of $S$. By lemma (\ref{invisible}) any causal curve from $r$ to $q$ is entirely contained in $\rm{Ext} (J^-(p))$. Since $S$ is a Cauchy surface any causal curve intersects $S$. Thus, for every $r \in J^-(q)\setminus J^-(p)$ that lies to the past of $S$, there exists $s \in M$ such that $s \in S\cap (J^-(q)\setminus J^-(p))$, i.e. $s \in S\setminus J^-(p)$.  
\end{proof} 
Let us proceed with the results that follow from theorem (\ref{non-empty}).
\begin{corollary}
Let $S$ be a Cauchy surface for the space-time $M$. No observer can access the entire initial data on $S$.
\end{corollary}
\begin{proof}
An observer moves on a time-like curve. For each point $p$ on the curve there exists a future point $q$. The result trivially follows from  theorem (\ref{non-empty}), since $S\cap (J^-(q)\setminus J^-(p))$, and $S\setminus J^-(p)$ are non-empty and invisible to  $p$.
\end{proof}
It is not necessary for an observer at $p$ to access the entire initial on a Cauchy surface $S$ to predict a future event $q$, since $S\setminus J^-(q)$ is invisible to $q$ itself. $q$ only receives signals from $J^-(q)$, and these signals are registered at $S \cap J^-(q)$. Therefore the initial data on $S \cap J^-(q)$ is sufficient to determine what happens at $q$. A natural question at this stage is whether an observer at $p$ can access the data on $S \cap J^-(q)$  to predict what happens at $q$.
\begin{corollary}
It is not possible for any observer to predict a future event with certainty.
\end{corollary}
\begin{proof}
Let $p<q$. By theorem (\ref{non-empty}) $S\cap (J^-(q)\setminus J^-(p))$ is non-empty and invisible to $p$. Thus a subset of $S \cap J^-(q)$ is invisible to $p$. Therefore an observer at $p$ cannot predict what happens at $q$ with certainty.
\end{proof}
An observer at space-time point $p$ can construct his prediction about what happens at a future point $q$, upon the initial data that is accessible to him, which is restricted to $S \cap J^-(p)$. There exists an uncertainty in his prediction, since the region $S\cap (J^-(q)\setminus J^-(p))$ is invisible to him, and the data in this region is required to determine what happens at   $q$ with certainty. Let us also investigate what happens as one attempts to predict further future.
\begin{corollary}
The uncertainty increases as one attempts to predict further future.
\end{corollary}
\begin{proof}
Let $p<q_1<q_2 \ldots <q_i<q_{i+1}<\ldots $. $S\cap (J^-(q_i)\setminus J^-(p))$ is non-empty and invisible to $p,q_1,q_2, \ldots q_{i-1}$. We have $ S\cap (J^-(q_i)\setminus J^-(p))=[S\cap (J^-(q_i)\setminus J^-(q_{i-1}))]\cup [S\cap (J^-(q_{i-1})\setminus J^-(q_{i-2}))]  \cup \ldots \cup[S\cap (J^-(q_1)\setminus J^-(p))]
$
Thus the uncertainty in predicting $q_{i+1}$ is larger than the uncertainty in predicting $q_i$.
\end{proof}
The classical uncertainty in prediction is generated by the region on a Cauchy surface that is invisible to an observer. As an observer attempts to predict further future this region enlarges, and the uncertainty increases.

The theorem (\ref{non-empty}) only applies to points that lie to the future of an observer. This can be generalised to include the whole space-time manifold.
\begin{theorem}
Let $p,q \in M$ such that $p\neq q$. An observer at $p$ can access the entire initial data on a Cauchy surface $S$ required to determine what happens at $q$, if and only if $q \in J^-(p)$.
\end{theorem}
\begin{proof}
We consider the three possibilities: $q \in J^+(p)$, $ q \in J^-(p) $, and $q \in M \setminus [J^+(p) \cup J^-(p)]=\rm{Ext} [J^+(p) \cup J^-(p)]$.
If $q \in J^+(p)$, we have proved that $S\cap (J^-(q)\setminus J^-(p))$ is non-empty and $p$ cannot predict what happens at $q$.

Let $ q \in J^-(p) $. One should access the data on $S \cap J^-(q)$ to determine $q$ with certainty. Since $J^-(q)\subset J^-(p)$, $[S \cap J^-(q)] \subset [S \cap J^-(p)]$. An observer at $p$ can determine what happens at $q$ with certainty.

Let $q \in \rm{Ext} [J^+(p) \cup J^-(p)]$. Let $r \in (J^-(q)\setminus J^-(p))$. One can proceed as in lemma (\ref{invisible}) to show that any causal curve from $r$ to $q$ is invisible to $p$. Let $S$ be a Cauchy surface. Again one can proceed as in theorem (\ref{non-empty}) to show that $S\cap (J^-(q)\setminus J^-(p))$ is non-empty and invisible to $p$. An observer at $p$ cannot determine what happens at $q$ with certainty.
\end{proof}

Though an observer cannot predict a future event, he can  determine a past event by using the initial data on a Cauchy surface, provided that he has an accurate theory. An observer moves on a time-like curve. Therefore a point at the future of an observer lies to his past after finite time. Then it can be determined with certainty. The point in determining a past event by using the data on a Cauchy surface whereas one could simply observe it, is that the former approach allows one to determine the state of a system without interacting with the system. In the original formulation of quantum mechanical uncertainty Heisenberg suggests a thought experiment where one measures the position of an electron by illuminating it with gamma rays.(see e.g.~\cite{riggs}) This results in a disturbance of the momentum of electron, i.e. a simultaneous determination of position and momentum is not possible. Heisenberg concluded that, this entailed the failure of classical causality at the quantum level~\cite{heisenberg}. If the quantum mechanical uncertainty was to be interpreted in terms of a simple measurement error and disturbance of another quantity as in its original form, one could argue that the classical causality still holds since it is possible to determine a past event with certainty without observing it. However according to the Copenhagen interpretation, the quantum mechanical uncertainty is an intrinsic feature of events. A physical quantity does not have a definite value unless it is measured. On the other hand, the classical uncertainty in predicting an event described in this work, cannot be interpreted as an intrinsic feature. It merely stems from the fact that one cannot access the entire initial data required to predict a future event. The classical uncertainty derived in this work and the usual quantum mechanical uncertainty are independent concepts which add up to bring a limitation to our ability to predict the future. However as far as past events are concerned, we are only limited by the intrinsic uncertainty imposed by quantum mechanics. 

It is convenient to elucidate the concept of classical uncertainty with a thought experiment. Let us assume that we release a particle in a gravitational field. We have a perfect classical theory to describe the dynamics of this particle in the gravitational field, so that we can predict the position, momentum or any relevant physical observable at any time with absolute accuracy; apart from the intrinsic quantum mechanical uncertainty. We make a prediction at $t=t_0$ to assign certain values to the physical observables at $t=t_b$. Let us envisage that a package of gamma rays were emitted from a distant galaxy a long time before $t_0$. These rays reach us at $t=t_a<t_b$ to interact with the particle and lead to deviations on the world line of the particle.  It is not possible for us to include the effect of gamma rays in our calculations at $t=t_0$, since these rays are \emph{invisible} to us at $t=t_0$ as discussed in lemma (\ref{invisible}). If there exists a Cauchy surface that is intersected by these rays, the relevant data on this Cauchy surface is also \emph{invisible} to us, as stated in theorem (\ref{non-empty}). Our predictions for the state of the particle at $t=t_b$ will be falsified. Since we cannot access the entire initial data that influences a future event, there exists a classical uncertainty in our prediction of the future in addition to the usual quantum mechanical uncertainty in the state of the particle.

\end{document}